\documentclass[letterpaper, 10 pt, conference]{ieeeconf}
\IEEEoverridecommandlockouts  
\let\labelindent\relax
\usepackage{enumitem}

\usepackage{mathtools,amsthm}

\usepackage{graphics} 
\usepackage{epsfig}
\usepackage{times} 
\usepackage{amsmath,amssymb,amsfonts}
\usepackage{algorithmic}
\usepackage{graphicx}
\usepackage{textcomp}
\usepackage{varioref}
\usepackage{import}
\usepackage{framed,xcolor}
\usepackage{color}
\usepackage{comment}
\usepackage{epstopdf}   
\usepackage{psfrag}
\usepackage{breqn}
\usepackage{float}
\usepackage{mathtools}
\usepackage{booktabs}
\usepackage{soul}
\usepackage{amsbsy}
\usepackage[bottom]{footmisc}
\usepackage{lineno}
\usepackage{cleveref}
\usepackage{microtype}
\usepackage{comment}
\newtheorem{theorem}{Theorem}

\newtheorem{Lemma}{Lemma}
\newtheorem{Definition}{Definition}

\newtheorem{Assumption}{Assumption}
\definecolor{arash}{rgb}{0.8,0.8,1}
\definecolor{seb}{rgb}{0.8,1,0.8}
\definecolor{seb2}{rgb}{0.5,.5,1}
\definecolor{arash2}{rgb}{0,.5,0}

\graphicspath{{Figures/}}
\newcommand{\vect}[1]{\ensuremath{\boldsymbol{\mathrm{#1}}}}

\title{\LARGE \bf Bias Correction in Deterministic Policy Gradient Using Robust MPC}
\author{Arash Bahari Kordabad, Hossein Nejatbakhsh Esfahani, Sebastien Gros
\thanks{The authors are with Department of Engineering Cybernetics, Norwegian University of Science and Technology (NTNU), Trondheim, Norway. E-mail:{\tt\small \{Arash.b.kordabad, hossein.n.esfahani, sebastien.gros\}@ntnu.no}}}
\begin{document}
\bstctlcite{IEEEexample:BSTcontrol}
\maketitle
\thispagestyle{empty}
\pagestyle{empty}
\begin{abstract}
In this paper, we discuss the deterministic policy gradient using the Actor-Critic methods based on the linear compatible advantage function approximator, where the input spaces are continuous. When the policy is restricted by hard constraints, the exploration may not be Centred or Isotropic (non-CI). As a result, the policy gradient estimation can be biased. We focus on constrained policies based on Model Predictive Control (MPC) schemes and to address the bias issue, we propose an approximate Robust MPC approach accounting for the exploration. The RMPC-based policy ensures that a Centered and Isotropic (CI) exploration is approximately feasible. A posterior projection is used to ensure its exact feasibility, we formally prove that this approach does not bias the gradient estimation.
\end{abstract}
\section{INTRODUCTION}
Reinforcement learning (RL) provides powerful tools for tackling Markov Decision Processes (MPDs) without depending on the probability distribution underlying the state transition~\cite{bertsekas2019reinforcement,sutton2018reinforcement}. RL methods attempt to enhance the closed-loop performance of a control policy deployed on the MDP, using observed realisation of the state transitions and of the corresponding stage cost. RL methods are usually either direct, based on an approximation of the optimal policy (e.g., deterministic and stochastic policy gradient methods~\cite{silver2014deterministic}) or indirect, based on an approximation of the action-value function (e.g., Q-learning). Unstructured function approximation techniques (e.g., Deep Neural Networks) are often used to carry these approximations. Unfortunately, the closed-loop behavior of such approximators can be challenging to analyze formally. In contrast, structured function approximations such as Model Predictive Control (MPC) schemes provide a formal framework to analyse the stability and feasibility of the closed-loop system~\cite{wabersich2018safe}. Recent research have focused on MPC-based policy approximation for RL~\cite{koller2018learning,gros2019data,zanon2020reinforcement,bahari2021reinforcement,gros2020reinforcement,nejatbakhsh2021reinforcement}.
\par For computational reasons, simple models are usually preferred in the MPC scheme. Hence, the MPC model often does not have the structure required to correctly capture the real system dynamics and stochasticity. As a result, while MPC can deliver a reasonable approximation of the optimal policy, it is usually suboptimal~\cite{MPCbook}. Choosing the MPC model parameters that maximise the closed-loop performance of the MPC scheme is a difficult problem, and the parameters that best fit the MPC model to the real system are not guaranteed to yield the best MPC policy \cite{gros2019data}. In \cite{gros2020reinforcement,gros2019data}, it is shown that adjusting not only the MPC model, but also the cost and constraints can be beneficial to achieve the best closed-loop performances, and RL is proposed as a possible approach to perform that adjustment in practice. In the presence of uncertainties and stochasticity, if constraints satisfaction is critical, Robust Model Predictive Control (RMPC) provides tools to ensure that the constraints are satisfied, and can be used in the RL context \cite{zanon2020safe}.

%The uncertainty may arise in different ways, such as additive disturbance, unknown model parameters, etc.

%\par MPC uses a model of the real system and, in general, delivers a suboptimal but useful approximation of the optimal policy~\cite{MPCbook}. MPC is often adopted for its ability to handle state and input constraints explicitly~\cite{mayne2014model}. 

\par Actor-Critic (AC) techniques combine the strong points of actor-only (policy search methods) and critic-only (e.g., Q-learning) methods~\cite{konda2000actor}. AC approaches are based on genuine optimality conditions of the closed-loop policy and typically deliver less noisy policy gradients than direct policy search. The deterministic policy gradient is built based on an approximation of the advantage function associated with the policy. To this end, a linear compatible advantage function approximator is a convenient choice, because it provides a correct policy gradient estimation with a given structure and a low number of parameters~\cite{silver2014deterministic}. For deterministic policies, exploration is required in order to estimate the corresponding policy gradient. % and discover directions in the policy parameters that can improve the closed-loop performance. 
In the presence of hard constraints, this exploration can be restricted. As a result, the exploration may become non-CI. In~\cite{Bias2020} it is shown that a linear compatible advantage function approximator can deliver an incorrect policy gradient estimation for a non-CI exploration.

\par In this paper, we propose to use a RMPC scheme that is robust with respect to a bounded disturbance of its first control input to enable the feasibility of a CI exploration. Because RMPC is computationally expensive, we use an inexpensive approximate RMPC instead, feasible to a first-order approximation. To ensure the feasibility of the exploration, a posterior projection technique is used. As a main result of this paper, we formally prove that the exploration resulting from RMPC scheme delivers an unbiased policy gradient estimation.

\par The paper is structured as follows. Section \ref{sec:back} provides background material on RL and details the bias problem. Section \ref{sec:RMPC} presents the RMPC-based approach that tackles the problem. For the sake of simplicity, we will consider a formulation robust with respect to the exploration only, while in practice the formulation can also be robust against model uncertainties and the stochasticity of the real system, as in \cite{zanon2020safe}. Section \ref{sec:FeasibilityForce} presents the projection approach required for nonlinear problems. Section \ref{sec:CORRPOLGRAD} describes the main theorem in the gradient bias correction using RMPC-based policy and proves that the resulting approach asymptotically yields a correct policy gradient. Section \ref{sec:sim} provides numerical examples of the method. Section \ref{sec:conc} delivers a conclusion.

\section{Background}\label{sec:back}
For a given MDP with continuous state-input space, a deterministic policy parametrized by $\vect \theta$ delivers an input $\vect a\in\mathbb{R}^m $ as a function of state $\vect s\in\mathbb{R}^n$ as,
%\begin{align}
$\vect \pi_{\vect \theta} (\vect s) : \mathbb{R}^n \rightarrow \mathbb{R}^m$. 
%\end{align}
If delivered by an MPC scheme, this policy is obtained as:
\begin{align}
\label{eq:Policy0}
\vect{\pi}_{\vect\theta}\left(\vect s\right) = \vect u_0^\star\left(\vect{s},\vect{\theta}\right),
\end{align}
where $\vect u_0^\star$ is the first element of the solution $\vect u^\star$ given by:
\begin{subequations}\label{eq:MPC}
\begin{align}
\min_{\vect{u},\vect x}&\quad V_{\vect\theta}(\vect x_{N}) +  \sum_{k=0}^{N-1} \gamma^k \ell_{\vect\theta}(\vect x_{k},\vect u_{k}), \label{eq:MPC:cost}\\
\mathrm{s.t.}&\quad \vect x_{k+1} = \vect f_{\vect\theta}\left(\vect x_{k},\vect u_{k}\right),\,\,\, \vect x_{0} = \vect s, \label{eq:Dynamics:MPC:Robust}\\
&\quad  \vect{h}_{\vect\theta}\left(\vect{x}_{k},\vect{u}_{k}\right) \leq 0, \label{eq:SafetyConstraints:MPC}\quad  \vect{h}_{\vect\theta}^{\mathrm f}\left(\vect{x}_{N}\right) \leq 0,
\end{align}
\end{subequations}
where $V_{\vect\theta}$ and $\ell_{\vect\theta}$ are the MPC terminal and stage costs, respectively. Function $\vect f_{\vect\theta}$ is the model dynamics and $\vect{h}_{\vect\theta}$ and $\vect{h}_{\vect\theta}^{\mathrm f}$ are the stage and terminal inequality constraints, respectively. Vector $\vect x=\{\vect x_{0,\ldots,N}\}$ is the predicted state trajectory and $\vect u=\{\vect u_{0,\ldots,N-1}\}$ is the input profile. State $\vect s$ is the current state of the system,  $N$ is the horizon length and $\gamma \in [0,1]$ is the discount factor. For the following theoretical developments, it will be useful to consider a single-shooting formulation of MPC \eqref{eq:MPC} resulting in a parametric Nonlinear Program (NLP):
\begin{subequations}\label{eq:Generic:NLP}
\begin{align}
\min_{\vect{u}}&\quad \Phi_{\vect{\theta}}(\vect s,\vect{u}), \\
\mathrm{s.t.}&\quad \vect{H}_{\vect{\theta}}\left(\vect s,\vect{u}\right) \leq 0, \label{eq:SafetyConstraints}
\end{align}
\end{subequations}
delivering the input profile of \eqref{eq:MPC} for all $\vect\theta,\vect s$ for some cost $\Phi_{\vect{\theta}}$ and inequality constraints $\vect{H}_{\vect{\theta}}$. We seek the policy parameters $\vect\theta$ that minimize the overall closed-loop cost $J$ of the policy ${\vect{ \pi}}_{\vect{\theta}}$ defined as follows:
\begin{align}
J(\vect{ \pi}_{\vect\theta}) = \mathbb{E}_{{\vect { \pi}_{\vect\theta}}}\left[\left. \sum_{k=0}^\infty\, \gamma^k L(\vect s_k,\vect a_k)\,\right|\, \vect a_k = \vect { \pi}_{\vect\theta}\left(\vect s_k\right)\, \right],
\end{align}
where $L(\vect s,\vect a) \in\mathbb{R}$ is the baseline stage cost evaluating the policy performance. It is shown in \cite{gros2019data} that using an MPC stage cost $\ell_{\vect\theta}$ different from the baseline stage cost $L$ can be beneficial when the MPC model is not exact. The expectation $\mathbb{E}_{{\vect { \pi}_{\vect\theta}}}$ is taken over the distribution of the Markov chain in closed-loop with the policy ${\vect{ \pi}}_{\vect{\theta}}$. The policy gradient for the deterministic policy ${\vect { \pi}}_{\vect{\theta}}$ is obtained as follows~\cite{silver2014deterministic}:
\begin{align}\label{eq:dj}
\nabla_{\vect{\theta}}\, J({\vect{ \pi}}_{\vect{\theta}})
&= \mathbb{E}_{\vect s}\left[\nabla_{\vect {\theta}} {\vect{ \pi}}_{\vect{\theta}}(\vect s)\, \nabla_{\vect a}A_{{\vect{ \pi}}_{\vect{\theta}}}(\vect s,\vect\pi_{\vect\theta}(\vect s))\right], %\nonumber    
\end{align}
where $A_{{\vect{ \pi}}_{\vect{\theta}}}(\vect s,\vect a)=Q_{\vect{\pi}_{\vect{\theta}}}(\vect s,\vect a) - V_{\vect{\pi}_{\vect{\theta}}}(\vect s)$ is the advantage function associated to  ${\vect{ \pi}}_{\vect{\theta}}$, and where $Q_{\vect{\pi}_{\vect{\theta}}}$ and $V_{\vect{\pi}_{\vect{\theta}}}$ are the action-value and value functions for the policy ${\vect{ \pi}}_{\vect{\theta}}$, respectively. In a non-episodic context, the expectation $\mathbb{E}_{\vect s}$ is taken over the steady-state distribution of the Markov chain. In an RL context, the advantage function $ A_{\vect{ \pi}_{\vect{\theta}}}$ must be approximated and evaluated from data. In the following, we label the advantage function approximation as $A^{\vect{w}}_{\vect{ \pi}_{\vect{\theta}}}$ with parameter vector $\vect w$. The corresponding estimation of the policy gradient in \eqref{eq:dj} reads as:
\begin{align}
\label{eq:dj:Approx}
\widehat{\nabla_{\vect{\theta}}\, J(\vect{ \pi}_{\vect{\theta}})} &= \mathbb{E}_{{\vect{ \pi}_{\vect{\theta}}}}\left[\nabla_{\vect{\theta}} \vect{ \pi}_{\vect{\theta}}(\vect s)\, \nabla_{\vect{a}}A^{\vect{w}}_{\vect{ \pi}_{\vect{\theta}}}(\vect s,\vect\pi_{\vect\theta}(\vect s))\right].
\end{align}
The following theorem provides the condition allowing one to replace the exact advantage $ A_{\vect{ \pi}_{\vect{\theta}}}$ in \eqref{eq:dj} by an approximation $A^{\vect{w}}_{\vect{ \pi}_{\vect{\theta}}}$, without affecting the policy gradient.
%\seb{If we need space we can probably dismiss that theorem and cut a few things to make the same point with less explanations.}\arash{OK. I will consider it as an option. However, I think without theory we need to explain other things and can't get too much space.}
\begin{theorem} \cite{silver2014deterministic}
If $A^{\vect{w}}_{\vect{ \pi}_{\vect{\theta}}}$ satisfies
\begin{enumerate}[label=\roman*.]
    \item $\nabla_{\vect a}A^{\vect{w}}_{\vect{ \pi}_{\vect{\theta}}}=\nabla_{\vect\theta}{\vect{\pi}_{\vect\theta}^\top \vect w}$,
    \item $\vect w$ minimizes the following mean-squared error:
\begin{align}\label{eq:A:TDEstimation2}
\vect{w} =\mathrm{arg}\min_{\vect{w}}\, \frac{1}{2}\mathbb{E}_{\vect s}\left[\left\|\nabla_{\vect a} A_{\vect{\pi}_{\vect{\theta}}} - \nabla_{\vect a}A^{\vect w}_{\vect{\pi}_{\vect{\theta}}}   \right\|^2\right],
\end{align}
\end{enumerate}
where the gradients are evaluated at $\vect a=\vect {\pi}_{\vect\theta}$, then we have:
\begin{align}
\label{eq:Correct:dJ}
\widehat{\nabla_{\vect{\theta}}\, J(\vect{ \pi}_{\vect{\theta}})}={\nabla_{\vect{\theta}}\, J(\vect{ \pi}_{\vect{\theta}})}.
\end{align}
\end{theorem}
\begin{proof}
See \cite{silver2014deterministic}.
\end{proof}
An advantage function approximator that achieves \eqref{eq:Correct:dJ} is labelled \textit{compatible}. A linear compatible advantage function approximator $A^{\vect{w}}_{\vect{ \pi}_{\vect{\theta}}}$, parametrized by $\vect w$ can read as \cite{silver2014deterministic}:
\begin{align}
\label{eq:CompatibleA}
A^{\vect{w}}_{\vect{ \pi}_{\vect{\theta}}}\left(\vect s,\vect{a}\right) =  \vect w^\top \nabla_{\vect{\theta}}\vect{ \pi}_{\vect{\theta}} \left(\vect{a} - \vect{ \pi}_{\vect{\theta}}\right).
\end{align}
It is well known that estimating $\nabla_{\vect{a}}A_{\vect{ \pi}_{\vect{\theta}}}$ directly is very difficult~\cite{silver2014deterministic}. As a surrogate to \eqref{eq:A:TDEstimation2}, the least-squares problem:
\begin{align}\label{eq:A:TDEstimation}
\vect{w} =\mathrm{arg}\min_{\vect{w}}\, \frac{1}{2}\mathbb{E}_{{\vect { \pi}_{\vect\theta}}}\left[\left(Q_{\vect{\pi}_{\vect{\theta}}} - \hat V_{\vect{\pi}_{\vect{\theta}}} - A^{\vect{w}}_{\vect{\pi}_{\vect{\theta}}}   \right)^2\right],
\end{align}
is used, where the value function estimation $\hat V_{\vect{\pi}_{\vect{\theta}}}\approx V_{\vect{\pi}_{\vect{\theta}}}$ is a baseline supporting the evaluation of $\vect{w}$. 
In order to obtain $\vect w$ from \eqref{eq:A:TDEstimation}, the input $\vect a$ applied to the real system must be different from the actual policy $\vect \pi_{\vect \theta}$, i.e. the input $\vect a$ applied to the real system should include some exploration in order to depart from the given policy $\vect \pi_{\vect \theta}$. One common choice of exploration is to add a random disturbance $\vect e$ to the policy as follows:
\begin{align}
\vect a=\vect \pi_{\vect \theta} (\vect s) +\vect e.
\end{align}
For the sake of clarity, we define hereafter a CI exploration.
\begin{Definition}
An exploration $\vect e$ is Centred and Isotropic (CI) if $\mathbb{E}_{\vect {e}}[\vect {e}]=0 $, and there exists a scalar $p$ such that, $
 \mathbb{E}_{\vect {e}}[\vect {e}\vect {e}^\top]=pI.
$
Otherwise it is non-CI.
\end{Definition}
\par Since the policy $\vect \pi_{\vect \theta} $ is subject to the hard constraints \eqref{eq:SafetyConstraints}, an arbitrary input $\vect a$ resulting from a random exploration $\vect e$ may not be feasible. Hence the exploration ought to be restricted such that it respects the constraints. A possible solution for this problem is, e.g., to use a projection of $\vect a$ on the feasible set of NLP \eqref{eq:Generic:NLP}. In the following we provide a definition for the projection operator.
\begin{Definition}
For an arbitrary input $\vect a$, the projection operator $P(\vect s,\vect a)$ is defined as follows:
\begin{subequations} \label{eq:projection}
\begin{align} 
P(\vect s,\vect a) &= \vect u^\perp_0,\\
\vect u^\perp &=\mathrm {arg}\min_{\vect{u}} \quad \frac{1}{2} \|\vect u_0-\vect a\|^2, \label{eq:Pcost}\\
&\qquad\quad\mathrm{s.t.}\quad \vect{H}_{\vect{\theta}}\left(\vect s,\vect{u}\right) \leq 0, \label{eq:Pconst}
\end{align}
\end{subequations}
where $\vect u^\perp_0$ is the first element of the input profile $\vect u^\perp_{0,\ldots,N-1}$ solution of \eqref{eq:Pcost}-\eqref{eq:Pconst}.
\end{Definition}
In particular, at a given state $\vect s$, the input $\vect a_\perp$ resulting from projecting the exploration is given by:
\begin{align}
\label{eq:Proj:a}
 \vect a_\perp = P\left(\vect s,\vect \pi_{\vect\theta}\left(\vect s\right)+\vect e\right).   
\end{align}
%More specifically, consider input $\vect a_\perp$ defined as:
%\begin{align}
%\vect{a_\perp}:=\vect u^\perp_0
%\end{align}
%where the input profile $\vect u^\perp_{0,\ldots,N-1}$ is solution of:%One can observe that $\vect{a_\perp}$ (together with the sequence $\vect u^\perp_{1,\ldots,N-1}$) is by construction a feasible input for the NLP \eqref{eq:Generic:NLP}. We will label the corresponding exploration as:
%\begin{align}
%\vect {e_\perp}=\vect {a_\perp}- \vect \pi_{\vect \theta} (\vect s).
%\end{align}
Then the projected exploration $\vect e_\perp$ is given by:
\begin{align}\label{eq:expl:prof}
\vect e_\perp = \vect a_\perp-\vect \pi_{\vect\theta}.
\end{align}
Unfortunately, even if the selected exploration $\vect e$ is CI, the projected exploration $\vect e_\perp$ may not be~\cite{Bias2020}. It is shown in \cite{Bias2020} that the linear compatible function approximator \eqref{eq:CompatibleA} using the fitting problem \eqref{eq:A:TDEstimation} delivers a correct estimated policy gradient \eqref{eq:dj:Approx}  only for a CI exploration. 

In this paper, we modify \eqref{eq:Generic:NLP} to find a policy $\vect{\hat \pi}_\theta\left(\vect s\right)$ for which a CI exploration is feasible. This policy $\vect{\hat \pi}_\theta\left(\vect s\right)$ is based on creating a small distance from the boundaries of the constraints so that a small CI exploration is feasible. To perform this modification, in the next section we will introduce an approximate RMPC scheme having a computational complexity similar to a standard MPC scheme. This RMPC scheme delivers a policy that can be disturbed with an additive perturbation in a given ball while keeping feasibility to a first-order approximation.

%We ought to stress again here that since the real system is possibly stochastic and the MPC model inaccurate, the policy resulting from the MPC scheme typically does not ensure that the constraints \eqref{eq:SafetyConstraints:MPC} are respected by the real system trajectories, even if no exploration is used. For the sake of simplicity, we present a RMPC scheme that accounts for the exploration only, while in practice, a guaranteed satisfaction of the constraints \eqref{eq:SafetyConstraints:MPC} requires that the RMPC also includes a (possibly conservative) representation of the model uncertainty and stochasticity of the real system as in e.g.~\cite{zanon2020safe}. The theory presented hereafter remains valid for that extension.

%The RMPC scheme will not necessarily ensure that the real system respects the constraints unless it is formulated so as to treat the model error and disturbances on top of the exploration. For the sake of simplicity, the RMPC we investigate here focuses on allowing a feasible CI exploration. To ensure the constraints satisfaction of the real system subject to the RMPC-based policy in the presence of stochasticity and model uncertainties, a further robustification of the RMPC scheme should be included as proposed in~\cite{zanon2020safe} in the context of Safe RL. 

\section{RMPC-based deterministic policy}\label{sec:RMPC}
In this section, we propose a modified policy $\hat{\vect{\pi}}_\theta$ based on an RMPC-scheme such that any input ${\vect {\hat a}}$ resulting from:
\begin{align}
\label{eq:explo}
\vect {\hat{a}} = \hat{\vect{\pi}}_\theta\left(\vect s\right) + \vect {\hat e},\quad \forall \vect {\hat e} \in B(0,\eta),
\end{align}
is feasible for the MPC \eqref{eq:MPC}, where $B(0,\eta)$ is a ball of radius $\eta$. For the sake of brevity, we consider in the following that the exploration $\vect {\hat e}$ is uniformly distributed in the ball $B(0,\eta)$. In that specific case, the exploration $\vect {\hat e}$ is CI with $p=\frac{1}{3} \eta^2$. To generate $\hat{\vect{\pi}}_\theta$ we tighten the inequality constraint \eqref{eq:SafetyConstraints} of NLP \eqref{eq:Generic:NLP} as follows:
\begin{subequations}
\label{eq:Modified:NLP}
\begin{align}
\min_{\vect{u}}&\quad \Phi_{\vect{\theta}}(\vect{s},\vect{u}), \\
\mathrm{s.t.}&\quad \vect{H}_{\vect{\theta}}\left(\vect{s},\vect{u}\right) + \vect\Delta_{\vect{\theta}}\left(\vect{s},\vect{u}\right)  \leq 0, \label{eq:SafetyConstraints2}
\end{align}
\end{subequations}
where $\vect\Delta_{\vect{\theta}}\left(\vect{s},\vect{u}\right) \geq 0$ is a back-off term added to ensure that the NLP \eqref{eq:Generic:NLP} is feasible for any additive perturbation $\vect {\hat e}\in B\left(0,\eta\right)$ of the input $\vect u_0$ obtained from \eqref{eq:Modified:NLP}. In general, evaluating $\vect\Delta_{\vect{\theta}}\left(\vect{s},\vect{u}\right)$ is difficult. To address this issue, we propose to compute $\vect\Delta_{\vect{\theta}}\left(\vect{s},\vect{u}\right)$ using a first-order approximation of the constraint \eqref{eq:SafetyConstraints}. More specifically, we will impose the approximated constraint:
\begin{align}
\label{eq:HVariation}
\vect{H}_{\vect{\theta}}\left(\vect s,\vect{\hat u}\right) &\approx  \vect{H}_{\vect{\theta}}\left(\vect{s},\vect{u}\right) + \left.\frac{\partial \vect{H}_{\vect{\theta}}}{\partial \vect u_0}\right|_{\vect u}%_0=\hat{\vect{\pi}}_{\vect\theta}}
\vect {\hat e}\, \leq 0,
\end{align}
where $\vect{\hat u}$ is the input profile resulting from perturbing $\vect{u}$ with the exploration $\vect {\hat e}\in B\left(0,\eta\right)$ in the first input $\vect u_0$. The following Lemma provides an explicit form for \eqref{eq:HVariation}.
\begin{Lemma} \label{Lem:Trivial}
Inequality \eqref{eq:HVariation} holds tightly for all $\vect {\hat e}\in B\left(0,\eta\right)$ if 
\begin{align}\label{eq:H:bound}
\vect{H}^i_{\vect{\theta}}\left(\vect{s},\vect{u}\right) + \left \| \frac{\partial \vect{H}^i_{\vect{\theta}}}{\partial \vect u_0}\bigg|_{\vect u} \right\| \eta\leq 0 \end{align}
holds, where $\vect{H}^i_{\vect{\theta}}$ is the $i^{\mathrm{th}}$ element of the vector $\vect{H}_{\vect{\theta}}$.
\end{Lemma}
\begin{proof}
The following inequality
\begin{align}
 \frac{\partial \vect{H}^i_{\vect{\theta}}}{\partial \vect u_0}\bigg|_{\vect u} \vect {\hat e} \leq 
%\end{align}
%has the solution:
%\begin{align}\label{eq:H:bound2}
\left \| \frac{\partial \vect{H}^i_{\vect{\theta}}}{\partial \vect u_0}\bigg|_{\vect u} \right\| \eta \,,\quad \forall \vect{\hat{e}}\in B(0,\eta), \end{align}
holds and is tight, where $\|.\|$ indicates an Euclidean norm. 
\end{proof}
The principles detailed above readily apply to MPC scheme \eqref{eq:MPC}. More specifically, an input disturbance $\vect {\hat e}$ yields:
%Then function $\vect\Delta_{\vect{\theta}}\left(\vect{u},\vect{s}\right) $ in Eq. \eqref{eq:SafetyConstraints2} can be defined as:
%\begin{align}
%\vect\Delta_{\vect{\theta}}\left(\vect{u},\vect{s}\right)_i =  \left \| \left(\frac{\partial %\vect{H}_{\vect{\theta}}}{\partial \vect x}\frac{\partial \vect x}{\partial \vect u_0}  + \frac{\partial %\vect{H}_{\vect{\theta}}}{\partial \vect u_0}\right)_i \right\|\eta
%\end{align}
%\seb{You should have introduced that before the NLP (3)}\arash{you mean before \eqref{eq:Modified:NLP}? but I think it is difficult to explain it there}\sebrep{I was talking about the MPC scheme (25a) and (25b) plus the unmodified constraints. It would make sense to explain early in the text where the NLP comes from. Then it would also make (25) easier to explain as a ``modified" scheme.} \arashrep{Sorry, I didn't get your point}\sebrep{I think you are making the explanations more complex they need to be by talking about a generic NLP for so long and introducing an actual MPC scheme that late. It could be easier to introduce the regular MPC scheme early (even before the NLP) so that the reader understands where the NLP is coming from.}
\begin{align}\label{eq:h:approx}
\vect{h}_{\vect{\theta}} &\approx  \vect{h}_{\vect{\theta}}\left(\vect{x}_k,\vect{u}_k\right) + \left(\frac{\partial \vect{h}_{\vect{\theta}}}{\partial \vect x_k}\frac{\partial \vect x_k}{\partial \vect u_0}  + \frac{\partial \vect{h}_{\vect{\theta}}}{\partial \vect u_0}\right)\vect {\hat e}, 
\end{align}
where the left hand side is evaluated of the perturbed trajectory and $\frac{\partial \vect x_k}{\partial \vect u_0}$ is obtained from the following linear dynamics:
\begin{align}
\label{eq:Sens}
\frac{\partial \vect x_k}{\partial \vect u_0} &= \left(\frac{\partial \vect{f}_{\vect{\theta}}}{\partial \vect x_{k-1}} \frac{\partial \vect x_{k-1}}{\partial \vect u_0} +  \frac{\partial \vect{f}_{\vect{\theta}}}{\partial \vect u_{0}}\right) \Big|_{\vect x_{k-1},\vect u_{k-1}},
\end{align}
with the initial condition $ \frac{\partial \vect x_0}{\partial \vect u_0}=0$. 

Imposing an arbitrary exploration radius $\eta$ may be infeasible for some state $\vect s$. To avoid this issue, we consider the radius as a decision variable $\nu \in [0,\bar\eta]$ whose optimal solution is $\eta$. We label $\bar\eta$ the maximum desired radius for the exploration. The RMPC-based policy $\vect{\hat \pi}_{\vect \theta}$ is then obtained as the first element of the input sequence given by:
\begin{subequations}
\label{eq:Modifed:MPC}
\begin{flalign}
\min_{\vect{u},\vect x,\nu}&\quad- w\nu+ V_{\vect\theta}(\vect x_{N}) +  \sum_{k=0}^{N-1}\gamma^k \ell_{\vect\theta}(\vect x_{k},\vect u_{k}),\\
\mathrm{s.t.}&\quad \vect x_{k+1} = \vect f_{\vect\theta}\left(\vect x_{k},\vect u_{k}\right),\,\,\, \vect x_{0} = \vect s, \label{eq:Dynamics:MPC:Robust2}\\
&\quad  \vect{h}_{\vect\theta}\left(\vect{x}_{0},\vect{u}_{0}\right)_i + \left\|\left( \frac{\partial \vect{h}_{\vect{\theta}}}{\partial \vect u_0}\right)_i\right\|\nu \leq 0, \\
&\quad  \vect{h}_{\vect\theta}\left(\vect{x}_{k},\vect{u}_{k}\right)_i + \left\|\left(\frac{\partial \vect{h}_{\vect{\theta}}}{\partial \vect x_k}\frac{\partial \vect x_k}{\partial \vect u_0}  \right)_i\right\|\nu \leq 0,\,\, k>0, \nonumber\\
&\quad  \vect{h}_{\vect\theta}^{\mathrm f}\left(\vect{x}_{N}\right)_i + \left\|\left(\frac{\partial \vect{h}^{\mathrm f}_{\vect{\theta}}}{\partial \vect x_N}\frac{\partial \vect x_N}{\partial \vect u_0}\right)_i \right\|\nu \leq 0 \\
&\quad  \frac{\partial \vect x_{k+1}}{\partial \vect u_0} = \frac{\partial \vect{f}_{\vect{\theta}}}{\partial \vect x_{k}} \frac{\partial \vect x_{k}}{\partial \vect u_0} ,\quad k = 2,\ldots,N-1, \label{eq:dynamics:RMPC}\\
&\quad   \frac{\partial \vect x_1}{\partial \vect u_0} = \frac{\partial \vect{f}_{\vect{\theta}}}{\partial \vect u_{0}},\quad
 0\leq \nu \leq \bar \eta, 
\end{flalign}
\end{subequations}
where $w$ is a positive constant weight, chosen large enough such that $\eta = \bar \eta $ when feasible. Index $i$ indicates the $i^{\mathrm{th}}$ element of the vectors.

One can observe that this RMPC scheme is feasible if the original MPC scheme \eqref{eq:MPC} is feasible. Indeed, the choice $\nu=0$ makes the RMPC and MPC schemes equivalent. It follows that the RMPC scheme \eqref{eq:Modifed:MPC} inherits the recursive feasibility of \eqref{eq:MPC}. We ought to stress again here that the recursive feasibility of \eqref{eq:MPC} may require the robust formulation to be extended to take the stochastic disturbances and model errors into account, as e.g. in~\cite{zanon2020safe}. We have omitted this aspect here for the sake of brevity and simplicity. The theory presented hereafter is applicable to that extension. Additionally, one ought to note that RMPC \eqref{eq:Modifed:MPC} is accounting for a disturbance on the initial input only. However, exploration is meant to be applied on all times. This could be reflected in the RMPC by accounting for a disturbance of the entire input profile, with minor modifications of the formulation. These modifications would, however, unnecessarily reduce the feasible domain of \eqref{eq:Modifed:MPC}. The proposed formulation arguably avoids that issue, and ensures feasibility via introducing the exploration radius as a decision variable in the NLP. Finally, a stabilizing feedback ought to be considered when forming the sensitivities \eqref{eq:Sens}, especially when the dynamics \eqref{eq:dynamics:RMPC} are unstable. This additional feedback is a classic tool to reduce the conservatism of the RMPC schemes. It is not presented here for the sake of brevity. 

Since a first-order approximation of the constraints is used when forming \eqref{eq:Modifed:MPC}, its solution may not ensure the feasibility of all exploration $\vect {\hat e}\in B\left(0,\eta\right)$. In the next section, we will address this problem with a posterior projection technique. We will show that this projection does not bias the policy gradient estimation.
\section{Ensuring feasibility} \label{sec:FeasibilityForce}
%Extra precautions are needed to ensure the feasibility of the exploration when the MPC scheme includes nonlinear constraints or dynamics. Indeed, 
Because we considered a first-order approximation of the constraints when forming the RMPC \eqref{eq:Modifed:MPC}, a posterior projection ought to be used to ensure the feasibility of the exploration. Using \eqref{eq:projection}, we apply the projection of $\vect{\hat{a}}$ on the feasible set as:
\begin{align}\label{eq:proj:a:hat}
 \vect {\hat a}_\perp = P\left(\vect s,\vect {\hat\pi}_{\vect\theta}+\vect {\hat e}\right).   
\end{align}
Using \eqref{eq:explo}, let us define the projection correction $\vect \epsilon$ as:
\begin{align}\label{eq:explo2}
\vect \epsilon\,:=\,\vect{\hat{a}_\perp}-\vect {\hat a}, 
\end{align}
and using \eqref{eq:expl:prof}, the feasible projected exploration $\vect {\hat e}_\perp$ can be written as follows:
\begin{align}\label{eq:explo3}
\vect {\hat e}_\perp :=  \vect{\hat{a}_\perp}-\vect {\hat\pi}_{\vect\theta} =\vect {\hat e}+\vect \epsilon.  
\end{align}
In the following we will show that the norm of $\vect \epsilon$ is in the order of $\eta^2$ for small enough $\bar\eta$. To this end, we make the following mild assumption for the constraints.
\begin{Assumption} \label{assum:H}
$\vect{H}_{\vect{\theta}}$ is a second order differentiable function and we have:
\begin{align}
\label{eq:BoundedSens}
   \forall i \,\,,\,\, \left\|\frac{\partial \vect{H}^i_{\vect{\theta}}}{\partial \vect u_0}\Big|_{\vect {\hat \pi}_{\vect\theta}}\right\|\neq 0. 
\end{align}
\end{Assumption}
Note that if the constraints satisfy Linear Independence Constraint Qualification (LICQ), then \eqref{eq:BoundedSens} is satisfied.
\begin{Lemma}\label{lemma:eps}
For the projection error $\vect \epsilon$ defined in \eqref{eq:explo2} and small enough $\bar\eta$, there exists a positive $\alpha$ such that:
\begin{align}\label{eq:error}
\|\vect \epsilon\| \leq \alpha \eta^2.
\end{align}
\end{Lemma}
\begin{proof} Let us define $\mathcal{H}_{\vect \theta}$ as, 
\begin{align}
\mathcal{H}_{\vect \theta} (\vect s,\vect u_0):=\vect {H}_{\vect \theta} (\vect s,\tilde{\vect u}),
\end{align}
where $\tilde{\vect u} := \left\{\vect u_0, \vect u_1^\perp,\ldots,\vect u_{N-1}^\perp\right\}$. We define $\mathcal{H}^i_{\vect \theta}$ as the $i^{\mathrm{th}}$ element of vector $\mathcal{H}_{\vect \theta}$. Consider the exploration described by its unitary direction $\vect v$, i.e. $\|\vect v\|=1$, and magnitude $\zeta \leq \eta$, i.e. $\vect {e} = \zeta\vect v$. We observe that:
\begin{align} \label{eq:H:Taylor}
\mathcal H_{\vect \theta}^i\left(\vect s,\hat{\vect\pi}_{\vect\theta}+\vect {\hat e}\right) \leq \mathcal H_{\vect \theta}^i\left(\vect s,\hat{\vect\pi}_{\vect\theta}\right) + (\mathcal{H}^i_{\vect \theta})'\vect {\hat e}  + R\left(\vect {\hat e}\right),
\end{align}
where $(\mathcal{H}^i_{\vect \theta})':=\frac{\partial \mathcal H_{\vect \theta}^i}{\partial \vect u_0}\big|_{\hat{\vect\pi}_{\vect\theta}}$. The inequality \eqref{eq:H:Taylor} holds for all $\vect {\hat e}\in B\left(0,\eta\right)$ for some continuous function $R\left(\vect {\hat e}\right)$, and there is a constant $c$ such that:
\begin{align}
\label{eq:Rbound}
\left|R\left(\vect {\hat e}\right)\right|\leq c\|\vect {\hat e}\|^2.    
\end{align}
Additionally: 
\begin{align}
\label{eq:IneqfromRMPC}
\mathcal H_{\vect \theta}^i\left(\vect s,\hat{\vect\pi}_{\vect\theta}\right) \leq -\eta \left\|(\mathcal{H}^i_{\vect \theta})'\right\|,
\end{align}
because $\left\{\hat{\vect\pi}_{\vect\theta}, \vect u_1^\perp,\ldots,\vect u_{N-1}^\perp\right\}$ is feasible for the RMPC scheme \eqref{eq:Modifed:MPC}. 
%hold by construction for all $\vect {e}$ in $B(0,\eta)$. It follows in particular that
%\begin{align}
 %- \left\|\frac{\partial \mathcal H_{\vect \theta}^i}{\partial \vect u_0}\right\| + \frac{\partial \mathcal H_{\vect \theta}^i}{\partial \vect u_0}\vect v\leq 0
%\end{align}
%\arash{right, it holds from cauchy-schwarz.}\seb{I'm actually not sure we need that...}holds for all $\vect v$. By continuity, for all $\vect {e}$ infeasible, there exists a $ t\in[0,1[$ such that
%\begin{align}
%\label{eq:Feasible:t}
%\mathcal H_{\vect \theta}^i\left(\hat{\vect\pi}_{\vect\theta}+t\vect {e},\vect s\right) \leq \mathcal H_{\vect \theta}^i\left(\hat{\vect\pi}_{\vect\theta},\vect s\right) + \left.\frac{\partial \mathcal H_{\vect \theta}^i}{\partial \vect u_0}\right|_{\hat{\vect\pi}_{\vect\theta}}t\vect {e}  + R\left(t\vect {e}\right) \leq 0
%\end{align}
%\arash{Yes. e.g., $t=0$ holds }
%holds.
Consider any sequence $\eta_k>0$ converging uniformly to $0$, and a corresponding sequence $t_k= \max(1-\alpha\eta_k,0)$ for some positive constant $\alpha$. One can readily observe that $t_k \vect {\hat e}\in B\left(0,\eta\right)$.  Additionally, by construction, there exists an index $k_0$ such that for all $k\geq k_0$,  $t_k=1-\alpha\eta_k$ holds. Using $t_k \vect {e}$ as the exploration in the right side of \eqref{eq:H:Taylor}, for $k\geq k_0$ we have:
\begin{align}
& \mathcal H_{\vect \theta}^i\left(\vect s ,\hat{\vect\pi}_{\vect\theta}\right) + (\mathcal{H}^i_{\vect \theta})'\vect {\hat e} \left(1-\alpha\eta_k\right) + R\left(\left(1-\alpha\eta_k\right)\vect {\hat e}\right) \leq \nonumber\\& -\eta_k \left\|(\mathcal{H}^i_{\vect \theta})'\right\| + (\mathcal{H}^i_{\vect \theta})'\zeta\vect v \left(1-\alpha\eta_k\right) +c\left(1-\alpha\eta_k\right)^2\zeta^2  \leq \nonumber \\& 
-\eta_k \left\|(\mathcal{H}^i_{\vect \theta})'\right\|+\left\|(\mathcal{H}^i_{\vect \theta})'\right\|\eta_k\left(1-\alpha\eta_k\right)+c\left(1-\alpha\eta_k\right)^2\zeta^2\nonumber \\& 
=-\left\|(\mathcal{H}^i_{\vect \theta})'\right\| \alpha\eta_k^2  + c\left(1-\alpha\eta_k\right)^2\eta_k^2 \leq 0, 
\end{align}
where the first inequality uses \eqref{eq:Rbound}  and \eqref{eq:IneqfromRMPC}. The second inequality is obtained from selecting $\zeta = \eta_k$ and using the Cauchy–Schwarz inequality.  Using Assumption \ref{assum:H}, the last inequality holds for $c\left\|(\mathcal{H}^i_{\vect \theta})'\right\|^{-1}\leq \alpha$. Therefore, $t_k \vect {\hat e}$ is a feasible exploration for \eqref{eq:Generic:NLP} and has a larger (or equal) error than the projection error. Then we have:
\begin{align}
&\|\vect \epsilon\| = \|\vect {\hat a}_\perp - \vect {\hat a} \| = \|\vect {\hat e}_\perp - \vect {\hat e} \|  \leq  \|t_k \vect {\hat e} - \vect {\hat e} \|=\nonumber\\&\|(1-t_k) \vect {\hat e}\|=\|
\alpha \eta_k \vect {\hat e}\|\leq \alpha \eta^2.
\quad \quad \quad \qedhere 
\end{align}
\end{proof}
%\end{comment}
%\arash{In the following I write the proof with another order}
%\begin{proof}
% We observe that:
%\begin{align}
%\mathcal H_{\vect \theta}^i\left(\hat{\vect\pi}_{\vect\theta}+\vect {e},\vect s\right) \leq \mathcal H_{\vect \theta}^i\left(\hat{\vect\pi}_{\vect\theta},\vect s\right) + \left.\frac{\partial \mathcal H_{\vect \theta}^i}{\partial \vect u_0}\right|_{\hat{\vect\pi}_{\vect\theta}}\vect {e}  + R\left(\vect {e}\right)
%\end{align}
%holds for all $\vect {e}\in B\left(0,\eta\right)$ for some continuous function $R\left(\vect {e}\right)$, and there is a constant $c$ such that $\left|R\left(\vect {e}\right)\right|\leq c\|\vect {e}\|^2$ on $B\left(0,\eta\right)$. Consider the following function:
% \begin{align}
%f_\alpha\left(\eta_k\right) := -\left\|\frac{\partial \mathcal H_{\vect \theta}^i}{\partial \vect u_0}\right\| \alpha +  c\left(1-\alpha\eta_k\right)^2 
% \end{align}
% where $\eta_k>0$ is any sequence converging uniformly to $0$, constant $\alpha>0$ is selected such that $\alpha \geq c\left\|\frac{\partial \mathcal H_{\vect \theta}^i}{\partial \vect u}\right\|^{-1}$. One can observe that $f_\alpha$ is a non-positive function. Hence we will show 
%\end{proof}

The following theorem provides some useful properties on the statistics of $\vect {\hat{e}_\perp}$.
\begin{theorem}\label{theorem:hat:e}
The projected exploration $\vect {\hat{e}_\perp}$ defined in \textnormal{(\ref{eq:proj:a:hat}-\ref{eq:explo3})}, for the policy resulting from RMPC \eqref{eq:Modifed:MPC}, has the following properties:
\begin{subequations}\label{eq:hat:e}
\begin{align} 
&\lim_{\bar\eta\rightarrow 0}\mathbb{E}_{\vect {\hat{e}_\perp}} [\vect {\hat{e}_\perp}]=0, \label{eq:hat:e1} \\ &\lim_{\bar\eta\rightarrow 0}\mathbb{E}_{\vect {\hat{e}_\perp}}\left[\frac{1}{\eta^2}\vect {\hat{e}_\perp}\vect {\hat{e}_\perp}^\top\right]=  \frac{1}{3} I, \label{eq:hat:e2} \\&\lim_{\bar\eta\rightarrow 0} \mathbb{E}_{\vect {\hat{e}_\perp}} \left[\frac{1}{\eta^2}\vect {\hat{e}_\perp} \xi (\vect{\hat{e}_\perp})\right]=0, \label{eq:hat:e3}
\end{align}
\end{subequations}
where $\eta$ is the solution of $\nu$ in the RMPC \eqref{eq:Modifed:MPC} and $\xi$ is any scalar function satisfying $|\xi(.)|\leq r\|.\|^2$ for some positive $r$.
\end{theorem}
\begin{proof}
We have
$\lim_{\bar\eta\rightarrow 0}\eta=0$,
%\end{align}
because $\eta\in[0,\bar\eta]$.
 Using Lemma \ref{lemma:eps}, we have:
\begin{align}\label{eq:e:zero}
\lim_{\bar\eta\rightarrow 0}\left\|\mathbb{E}[\vect {\epsilon}]\right\|&\leq\lim_{\bar\eta\rightarrow 0}\mathbb{E}\left[\left\|\vect {\epsilon}\right\|\right]\leq\lim_{\bar\eta\rightarrow 0} \alpha\eta^2=0 \nonumber\\&\Rightarrow \lim_{\bar\eta\rightarrow 0}\mathbb{E}[\vect {\epsilon}]=0.
\end{align}
Taking the expectation from \eqref{eq:explo3} and using that the exploration $\vect {\hat e}$ is CI, we have:
\begin{align}
\lim_{\bar\eta\rightarrow 0}\mathbb{E}[\vect {\hat{e}_\perp}]=\lim_{\bar\eta\rightarrow 0}\Big(\mathbb{E}[\vect {\hat e}]+\mathbb{E}[\vect {\epsilon}]\Big)=0.
\end{align}
Using \eqref{eq:explo3}, the second moment can be written as follows:
\begin{align}
&\lim_{\bar\eta\rightarrow 0}\mathbb{E}\left[\frac{1}{\eta^2}\vect {\hat{e}_\perp}\vect {\hat{e}_\perp}^\top\right]= \lim_{\bar\eta\rightarrow 0}\Bigg(\mathbb{E}\left[\frac{1}{\eta^2}\vect {\hat e}\vect {\hat e}^\top\right]+\Bigg.\nonumber \\&\qquad\Bigg.\mathbb{E}\left[\frac{1}{\eta^2} \vect {\epsilon}\vect {\hat e}^\top\right]+\mathbb{E}\left[\frac{1}{\eta^2} \vect {\hat e}\vect {\epsilon}^\top\right]+\mathbb{E}\left[\frac{1}{\eta^2} \vect {\epsilon}\vect {\epsilon}^\top\right]\Bigg).
\end{align}
%\seb{Here you equivocate between $\bar\eta$ and $\eta$... they are not necessarily the same and the first term may be affected...}\arash{The first term is constant 1/3I}\seb{Please read my comment.}
For the first term we use that the exploration $\vect {\hat e}$ is CI with $p=\frac{1}{3}\eta^2 I$, i.e:
\begin{align}
    \mathbb{E}[\vect {\hat e}\vect {\hat e}^\top]=\frac{1}{3} \eta^2I
    \Rightarrow \mathbb{E}\left[\frac{1}{\eta^2}\vect {\hat e}\vect {\hat e}^\top\right]=\frac{1}{3} I.
\end{align}
Using \eqref{eq:error}, for the second term we have:
\begin{align}
&\lim_{\bar\eta\rightarrow 0}\left\|\mathbb{E}\left[\frac{1}{\eta^2} \vect {\epsilon}\vect {\hat e}^\top\right]\right\|\leq
\lim_{\bar\eta\rightarrow 0}\mathbb{E}\left[\frac{1}{\eta^2} \|\vect {\epsilon}\|\|\vect {\hat e}\|\right]\leq\nonumber \\
&\qquad\qquad \lim_{\bar\eta\rightarrow 0}{\alpha}\mathbb{E}[\|\vect {\hat e}\|]=0  \Rightarrow \lim_{\bar\eta\rightarrow 0} \mathbb{E}\left[\frac{1}{\eta^2} \vect {\epsilon}\vect {\hat e}^\top\right]=0.
\end{align}
The third term will vanish in the similar way and for the forth term we can write:
\begin{align}
&\lim_{\bar\eta\rightarrow 0}\left\|\mathbb{E}\left[\frac{1}{\eta^2} \vect {\epsilon}\vect {\epsilon}^\top\right]\right\|\leq
\lim_{\bar\eta\rightarrow 0}\mathbb{E}\left[\frac{1}{\eta^2} \|\vect {\epsilon}\|\|\vect {\epsilon}\|\right]\leq \\
&\qquad \lim_{\bar\eta\rightarrow 0}{\alpha}\eta^2
\leq \lim_{\bar\eta\rightarrow 0}{\alpha}\bar\eta^2=0 \Rightarrow \lim_{\bar\eta\rightarrow 0} \mathbb{E}\left[\frac{1}{\eta^2} \vect {\epsilon}\vect {\epsilon}^\top\right]=0.\nonumber
\end{align}
Then, they deliver \eqref{eq:hat:e2}. Finally for \eqref{eq:hat:e3}, we have:
\begin{align}
\|\vect {\hat{e}}_\perp\|=\|\vect {\hat{e}}+\vect \epsilon\|\leq\|\vect {\hat{e}}\|+\|\vect \epsilon\|\leq\eta+\alpha \eta^2.
\end{align}
Then:
\begin{align}
    &\lim_{\bar\eta\rightarrow 0} \left\|\mathbb{E} \left[\frac{1}{\eta^2}\vect {\hat{e}_\perp} \xi (\vect {\hat{e}_\perp})\right]\right\|\leq \lim_{\bar\eta\rightarrow 0} \mathbb{E} \left[\frac{1}{\eta^2}\|\vect {\hat{e}_\perp}\| |\xi (\vect {\hat{e}_\perp})|\right]
    \nonumber\\&\qquad\leq  \lim_{\bar\eta\rightarrow 0} \mathbb{E} \left[\frac{r}{\eta^2}\|\vect {\hat{e}_\perp}\|^3\right]\leq  \lim_{\bar\eta\rightarrow 0} r\eta(1+\alpha \eta)^3=0,
\end{align}
which delivers \eqref{eq:hat:e3}. 
\end{proof}
\section{Corrected Policy Gradient} \label{sec:CORRPOLGRAD}
In this section, we will show that the robust policy $\vect{\hat \pi}_{\vect{\theta}}$ delivers the true gradient as $\bar\eta\rightarrow0$. Indeed, the deterministic policy gradient method uses ``small" exploration and all results are valid in the sense of $\bar\eta\rightarrow 0$. We propose the compatible advantage function:
\begin{align}\label{eq:CorrectedA}
A^{\vect{w}}_{\vect{\hat \pi}_{\vect{\theta}}}\left(\vect s,\vect{\hat{a}}_\perp\right) =  \frac{\bar\eta^2}{\eta^2}\vect w^\top \nabla_{\vect{\theta}}\vect{\hat \pi}_{\vect{\theta}} \left(\vect{\hat{a}}_\perp - \vect{\hat \pi}_{\vect{\theta}}\right),
\end{align}
where the factor $\frac{\bar\eta^2}{\eta^2}$ is required to account for the varying exploration radius $\eta$ and $\vect w$ is obtained as follows:
\begin{align}\label{eq:A:LS:corrected}
\vect{w} =\mathrm{arg}\min_{\vect{w}}\, \frac{1}{2}\mathbb{E}_{\vect{\hat \pi}_{\vect{\theta}},\vect {\hat{e}_\perp}}\left[\frac{1}{\bar\eta^2} \left(Q_{\vect{\hat\pi}_{\vect{\theta}}} - \hat V_{\vect{\hat\pi}_{\vect{\theta}}} - A^{\vect{w}}_{\vect{\hat\pi}_{\vect{\theta}}}   \right)^2\right]
\end{align}
where $\mathbb{E}_{\vect{ \hat \pi}_{\vect{\theta}},\vect {\hat{e}_\perp}}=\mathbb{E}_{\vect{\hat \pi}_{\vect{\theta}}}[\mathbb{E}_{\vect {\hat{e}_\perp}}[.|\vect s]]$ and $\bar\eta^{-2}$ is introduced such that \eqref{eq:A:LS:corrected} remains well-posed for $\bar\eta\rightarrow 0$.
\begin{Assumption} \label{assump}
$Q_{\vect{\hat\pi}_{\vect{\theta}}}$ is analytic and at least twice differentiable for almost every feasible $\vect s$ and $\nabla^2_{\vect {a}} Q_{\vect{\hat \pi}_{\vect{\theta}}}$ is bounded.
\end{Assumption}
Assumption \ref{assump} is usually satisfied in practice, as $Q_{\vect{\hat\pi}_{\vect{\theta}}}$ tends to be at least piecewise smooth for the many problems based on continuous state-input spaces. This assumption can be relaxed, but it requires more technical developments.
\begin{theorem}  The RMPC-based policy gradient estimation using the compatible advantage function in \eqref{eq:CorrectedA} with $\vect w$ given by \eqref{eq:A:LS:corrected} asymptotically converges to exact gradient, i.e.:
\begin{align}\label{eq:GradError}
&\lim_{\bar\eta\rightarrow 0}\widehat{\nabla_{\vect{\theta}}\, J(\vect{\hat \pi}_{\vect{\theta}})} ={\nabla_{\vect{\theta}}\, J(\vect{ \hat \pi}_{\vect{\theta}})}.
\end{align}
\end{theorem}
\begin{proof}
The solution of \eqref{eq:A:LS:corrected} is given by:
\begin{align}
\label{eq:w:sol}
\mathbb{E}_{\vect{ \pi}_{\vect{\theta}},\vect {\hat{e}_\perp}}\left[\frac{1}{\eta^2}\nabla_{\vect{\theta}}\vect{\hat\pi}_{\vect{\theta}} \vect {\hat{e}_\perp}\left(Q_{\vect{\hat\pi}_{\vect{\theta}}} - \hat V_{\vect{\hat\pi}_{\vect{\theta}}} - A^{\vect{w}}_{\vect{\hat\pi}_{\vect{\theta}}}   \right) \right] = 0.
\end{align}
Using Assumption~\ref{assump}, the Taylor expansions of $Q_{\vect{\hat\pi}_{\vect{\theta}}}$ and $A^{\vect{w}}_{\vect{ \hat\pi}_{\vect{\theta}}}$ are valid almost everywhere. They read as:
\begin{align}\label{eq:Q:expansion}
Q_{\vect{\hat\pi}_{\vect{\theta}}}\left(\vect s,\vect {\hat a}_\perp\right)
&= V_{\vect{\hat\pi}_{\vect{\theta}}}\left(\vect s\right) + \nabla_{\vect {a}} A_{\vect{\hat\pi}_{\vect{\theta}}}\left(\vect s,\vect{\hat\pi}_{\vect\theta}(\vect s)\right)^\top\vect {\hat{e}_\perp} + \vect\xi, \nonumber\\
A^{\vect{w}}_{\vect{\hat \pi}_{\vect{\theta}}}\left(\vect s,\vect{\hat{a}}_\perp\right)&= \nabla_{\vect {a}} A^{\vect{w}}_{\vect{\hat \pi}_{\vect{\theta}}}\left(\vect s,\vect{\hat\pi}_{\vect \theta}(\vect s)\right)^\top \vect {\hat{e}_\perp}, 
\end{align}
where $\vect\xi$ is the second-order remainder of the Taylor expansion of $Q_{\vect{\hat\pi}_{\vect{\theta}}}$ at $\vect {\hat{e}_\perp}=0$ and the identity $\nabla_{\vect {a}} Q_{\vect{\hat\pi}_{\vect{\theta}}} = \nabla_{\vect {a}} A_{\vect{\hat\pi}_{\vect{\theta}}}$ was used. Using Assumption~\ref{assump}, $\vect\xi$ is of order $\|\vect {\hat{e}}_\perp\|^2$ for almost every feasible $\vect s$. 
By substitution of \eqref{eq:Q:expansion} in \eqref{eq:w:sol}, we have:
\begin{align}
\label{eq:Q:Fitting:Bias}
&\mathbb{E}_{\vect s, \vect {\hat{e}_\perp}} \left[\frac{1}{\eta^2}\nabla_{\vect{\theta}}\vect{\hat\pi}_{\vect{\theta}}{\vect {\hat{e}_\perp}}{\vect {\hat{e}_\perp}^\top}\left(\nabla_{\vect{a}}A_{\vect{\hat\pi}_{\vect{\theta}}}-\nabla_{\vect{a}}A^{\vect{w}}_{\vect{\hat\pi}_{\vect{\theta}}}\right) \right]+\nonumber\\
&\mathbb{E}_{\vect s,\vect {\hat{e}_\perp}}\left[ \frac{1}{\eta^2}\nabla_{\vect{\theta}}\vect{\hat\pi}_{\vect{\theta}}{\vect {\hat{e}_\perp}}\vect\xi\right] +\nonumber\\ & \mathbb{E}_{\vect s,\vect {\hat{e}_\perp}}\left[ \frac{1}{\eta^2}\nabla_{\vect{\theta}}\vect{\hat\pi}_{\vect{\theta}}\vect {\hat{e}_\perp}\left(V_{\vect{\hat\pi}_{\vect{\theta}}} -\hat V_{\vect{\hat\pi}_{\vect{\theta}}}\right) \right] = 0.
\end{align}
Using Theorem \ref{theorem:hat:e}, the second and third terms will be vanish in the sense of $\bar\eta\rightarrow 0$ and the first term will be: 
\begin{align}
\label{eq:Q:Fitting:Solution}
\lim_{\bar\eta\rightarrow 0} \mathbb{E}_{\vect s}\left[\nabla_{\vect{\theta}}\vect{\hat\pi}_{\vect{\theta}}\left(\nabla_{\vect{a}}A_{\vect{\hat\pi}_{\vect{\theta}}}-\nabla_{\vect{a}}A^{\vect{w}}_{\vect{\hat\pi}_{\vect{\theta}}} \right) \right]=0, 
\end{align}
which delivers \eqref{eq:GradError}.\end{proof}
In addition, under mild conditions, the RMPC-based policy resulting from \eqref{eq:Modifed:MPC} converges to the main MPC-based policy resulting from \eqref{eq:MPC} as $\bar\eta\rightarrow 0$. For the sake of brevity, we do not formalise this statement here.
\section{Numerical Simulation}\label{sec:sim}
In this section, we propose two numerical examples in order to illustrate the theoretical developments. The first example directly compares the MPC-based policy and the RMPC-based policy and the optimal policy with a nonlinear constraint. We consider the deterministic scalar MDP $s^+=s+a$ with stage cost $L(s,a)=s^2+a^2$ , constraint $s^2+5a^2 \leq 1$ and discount factor $\gamma=0.9$. Then we use the following MPC scheme to extract the approximated policy:
\begin{subequations}
\begin{align}
\min_{\vect{x},\vect{u}} \quad &  x^2_{N}+ \sum_{k=0}^{N-1} \gamma^k(\theta x^2_{k}+u^2_{k}), \\
s.t.\quad & x_{k+1}=x_k+u_k\, , \,
x_k^2+5u_k^2 \leq 1 \, , \, x_0=s,
\end{align}
\end{subequations}
then $\pi_\theta(s)=u^\star_0(s)$ obtained from the first element of the input solution. We can build RMPC-scheme according \eqref{eq:Modifed:MPC} and extract the modified policy $\hat \pi_\theta(s)$. Fig.\ref{fig1} (top) illustrates the posterior projected error $\|\epsilon\|$ and the approximated feasible radius $\eta$ for $\bar\eta=0.05$. As it can be seen, e.g., a fixed radius exploration with $\eta=0.05$ may be infeasible at $s=\pm1$. Fig.\ref{fig1} (bottom) compares these policies with the MDP optimal policy. This simple example shows that the RMPC-based policy makes a distance with the feasible bound to guarantee the feasibility of the exploration in both directions. While classic-MPC is on the feasible set bound, a feasible exploration should only be in one direction.
\begin{figure}[ht]
\centering
\includegraphics[width=0.45\textwidth]{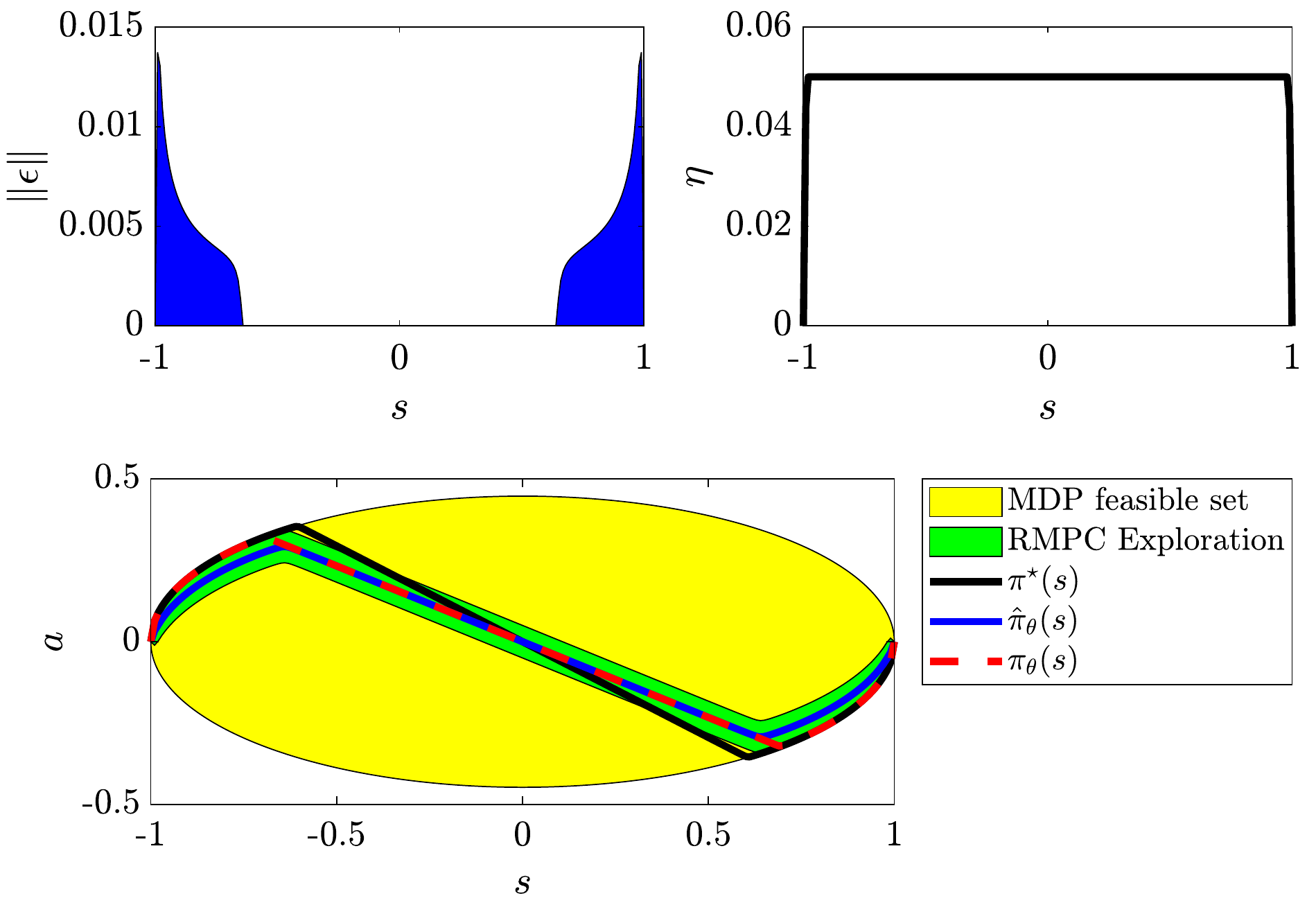}
\caption{Top-left: Blue region shows norm of the posterior projected error $\|\epsilon\|$. Top-right: The approximated feasible radius $\eta$. Bottom: The optimal policy and parametrized policies from MPC and RMPC scheme for $\theta=0.5$. }
\label{fig1}
\end{figure}
\par The second example compares the gradient of the RMPC-based and MPC-based policies with the true policy gradient. We consider linear scalar dynamics $
s^+=0.97s+0.1a+d
$
where $d\sim \mathcal{U}(-10^{-3},10^{3})$ is a scalar uniform noise. RL stage cost is
$
L(s,a)=20(s-0.5)^2+(a-2)^2
$
with $\gamma=0.9$. The policy is extracted from the following MPC scheme:
\begin{subequations}
\begin{align}
\min_u &\quad \sum_{k=0}^{50} \gamma^k (10(x_k-{1}/{3})^2+(u_k-u_{\mathrm{ref}}(\theta))^2), \\
\mathrm{s.t.} &\quad x_{k+1}=0.97x_{k}+0.1u_{k},\quad x_0=s, \\
&\quad u_k \leq \theta, \label{eq:cons:uref}
\end{align}
\end{subequations}
%\begin{subequations}
%\begin{align}
%\min_u &\quad \sum_{k=0}^{50} \gamma^k \left(10\left(x_k-\frac{1}{3}\right)^2+(u_k-u_{\mathrm{ref}}(\theta))^2\right) \\
%\mathrm{s.t.} &\quad x_{k+1}=0.97x_{k}+0.1u_{k}\quad ,\quad x_0=s \\
%&\quad u_k \leq \theta \label{eq:cons:uref}
%\end{align}
%\end{subequations}
where $u_{\mathrm{ref}}(\theta)=0.2-\theta$. The initial RL parameter $\theta = 0.1$ is selected. The MPC policy can be adjusted by increasing the input bound in \eqref{eq:cons:uref} and raising the input reference $u_{\mathrm{ref}}$. However, raising the input bound in \eqref{eq:cons:uref} by increasing $\theta$ results in decreasing the input reference $u_{\mathrm{ref}}$, such that these terms are in contradiction to find the optimal policy. Fig. \ref{fig3} shows the policy gradient over the RL iterations. The red (dashed) curve is the outcome of learning from the classic MPC, while the blue (solid) curve is the one from the RMPC. As it can be seen, the RMPC gradient $\widehat{\nabla_{\vect{\theta}}\,J(\vect{\hat \pi}_{\vect{\theta}})}$ delivers a very close gradient to the true gradient  $\nabla_{\vect{\theta}}\,J(\vect{\pi}_{\vect{\theta}})$. However, the MPC policy gradient $\widehat{\nabla_{\vect{\theta}}\, J(\vect{ \pi}_{\vect{\theta}})}$ has an obvious bias in both cases. Note that the closed-loop performance loss from this bias issue is not necessarily large for this example.
\begin{figure}[ht]
\centering
\includegraphics[width=0.4\textwidth]{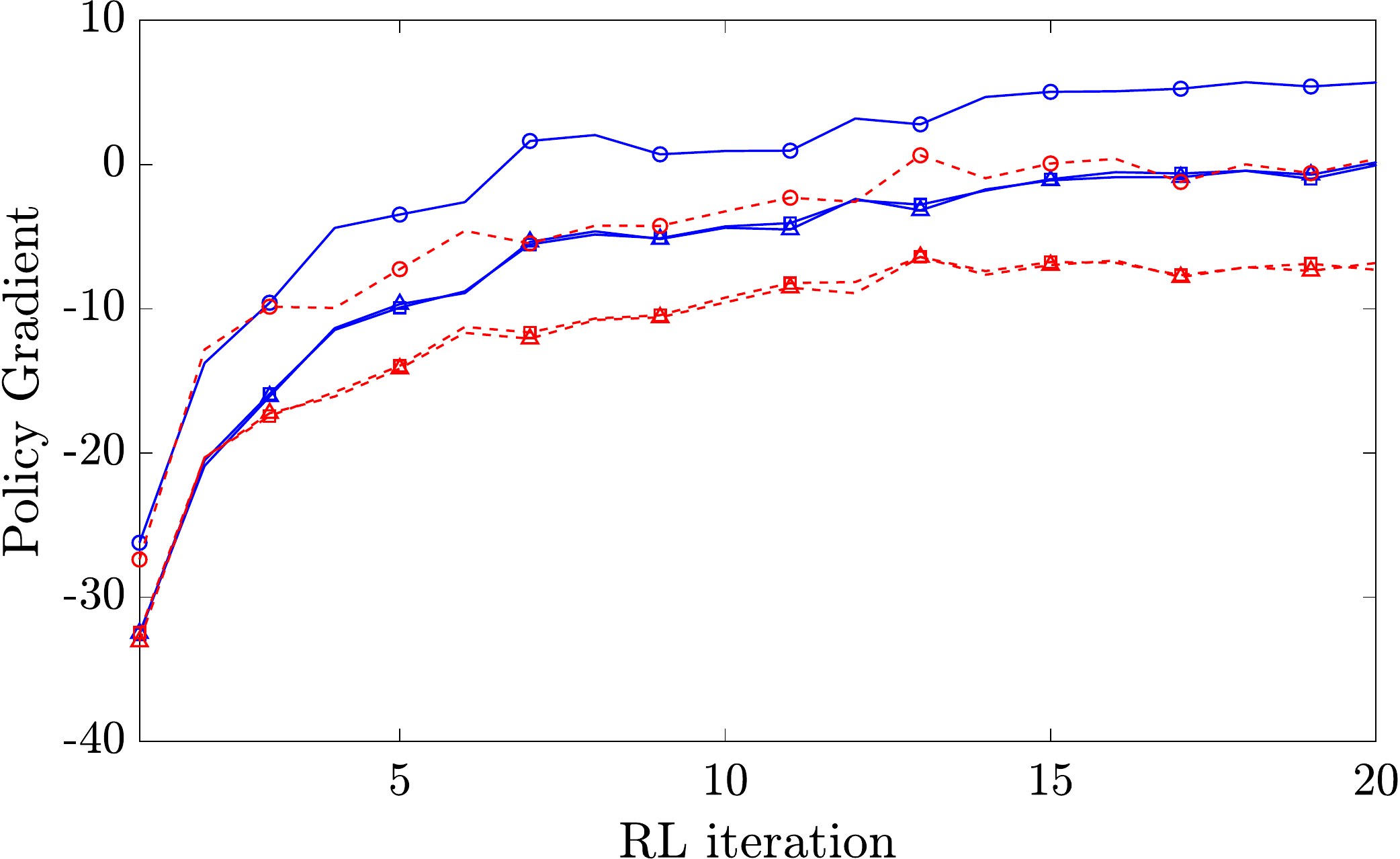}
\caption{ The policy gradients over the RL iterations. The outcome of learning using policy gradients from MPC (red-dashed) and RMPC scheme (blue-solid). ($\circ$ : $\widehat{\nabla_{\vect{\theta}}\,J(\vect{\pi}_{\vect{\theta}})}$. $\square$ :  $\nabla_{\vect{\theta}}\, J(\vect{ \pi}_{\vect{\theta}})$. $\vartriangle$ :  $\widehat{\nabla_{\vect{\theta}}\, J(\vect{\hat \pi}_{\vect{\theta}})}$.}
\label{fig3}
\end{figure}
\par A more complex example demonstrating the theory on a nonlinear example would be useful. For the sake of brevity, such an example will be considered in the future.
\section{CONCLUSION}\label{sec:conc}
This paper presented the AC approach using a linear compatible advantage function approximation for the MPC-based deterministic policies. When the policy is restricted by hard constraints, the exploration may be non-CI and delivers a bias in the policy gradient. We proposed RMPC using constraint tightening to provide an approximated feasible and CI exploration. A posterior projection is used to ensure feasibility and formally we showed that the RMPC-based policy gradient converges to the true policy gradient for a small enough radius of exploration. 
\bibliographystyle{IEEEtran}
\bibliography{Bias}
\end{document}